\documentclass[11pt]{article}
\usepackage{amsfonts}
\usepackage{amsmath}
\usepackage{latexsym}

\usepackage{amsthm}
\usepackage[bookmarks]{hyperref}
\usepackage{comment}
\usepackage{color}
\usepackage{soul}
\usepackage{tikz}
\usetikzlibrary{shapes.callouts,decorations.text}
 \usetikzlibrary{calc}
\usetikzlibrary{intersections,positioning,arrows,shapes,decorations.pathmorphing,decorations.pathreplacing}
\usetikzlibrary{shapes,calc,trees,positioning,through,fadings}

\tikzset{iNode/.style={draw=myblue, circle}}
\tikzset{fNode/.style={draw=mygreen, rectangle}}
\definecolor{myblue}{RGB}{80,80,160}
\definecolor{mygreen}{RGB}{80,160,80}

\def\andrei#1{\sethlcolor{lightblue}\hl{\texttt{#1}}}

\definecolor{lightred}{rgb}{1, .60, 0.60}
\definecolor{lightblue}{rgb}{.60,.60,1}
\definecolor{gold}{HTML}{FDD017}

\def\nottoobig#1{{\hbox{$\left#1\vcenter to1.111\ht\strutbox{}\right.\n@space$}}}

\newtheorem{theorem}{Theorem}[section]

\newtheorem{lemma}[theorem]{Lemma}
\newtheorem{claim}[theorem]{Claim}

\newtheorem{definition}[theorem]{Definition}

\theoremstyle{remark}
\newtheorem{remark}{Remark}
\newtheorem{example}[theorem]{Example}

\usepackage{epsfig}
\usepackage{amsmath}
\usepackage{amsfonts}

\newcommand{\poly}{{\rm poly}}

\def\nottoobig#1{{\hbox{$\left#1\vcenter
to1.111\ht\strutbox{}\right.\n@space$}}}

\newcommand{\ie}{$\mbox{i.e.}$}

\newlength{\filength}
\newsavebox{\gcbox}
\sbox{\gcbox}{\framebox[\filength]{\rule{0ex}{2ex}}}

\newcommand{\singlespacing}{\let\CS=
\@currsize\renewcommand{\baselinestretch}{1}\tiny\CS}
\newcommand{\singlespacingplus}{\let\CS=
\@currsize\renewcommand{\baselinestretch}{1.25}\tiny\CS}
\newcommand{\doublespacing}{\let\CS=
\@currsize\renewcommand{\baselinestretch}{1.75}\tiny\CS}
\newcommand{\draftspacing}{\let\CS=
\@currsize\renewcommand{\baselinestretch}{2.0}\tiny\CS}

\def\zo{\{0,1\}}

\newcommand{\leqp}{\leq^{+}}
\newcommand{\geqp}{\geq^{+}}
\newcommand{\eqp}{=^{+}}
\newcommand{\rt}{{\cal{R}}}
\newcommand{\ttt}{t_{x,y}}
\newcommand{\ttx}{\tilde{t}}
\newcommand{\ffx}{\tilde{f}}
\newcommand{\fff}{f_{x,y}}

\makeatletter
\def\@listI{\leftmargin\leftmargini \parsep 4.5pt plus 1pt minus 1pt\topsep6pt plus 2pt minus 2pt \itemsep  2pt plus 2pt minus 1pt}

\let\@listi\@listI
\@listi
\makeatother

\author{ 
{Andrei Romashchenko\/}
\thanks{Laboratoire d'Informatique, de Robotique et de Micro\'electronique de Montpellier (LIRMM) \& CNRS,
email:~\texttt{andrei.romashchenko@lirmm.fr}}
{\quad Marius Zimand\/}
\thanks{  Department of Computer and Information Sciences, Towson University,
Baltimore, MD. \texttt{http://orion.towson.edu/\~{ }mzimand}. Partially supported by NSF grant CCF 1811729.}}

\title{On a conditional inequality in Kolmogorov complexity and its applications in communication complexity}
\date{}
\begin{document}

\maketitle
\begin{abstract}
Romashchenko and Zimand~\cite{rom-zim:c:mutualinfo} have shown that if we partition the set of pairs $(x,y)$ of $n$-bit strings into combinatorial rectangles, then $I(x:y) \geq I(x:y \mid t(x,y)) - O(\log n)$, where $I$ denotes mutual information in the Kolmogorov complexity sense, and $t(x,y)$ is the rectangle containing $(x,y)$. We observe that this inequality can be extended to coverings with rectangles which may overlap. The new inequality essentially states that in case of a covering with combinatorial rectangles, 
 $I(x:y) \geq I(x:y \mid t(x,y)) - \log \rho - O(\log n)$, where $t(x,y)$ is any rectangle containing $(x,y)$ and $\rho$ is the thickness of the covering, which is the maximum number of rectangles that overlap.  We discuss applications to communication complexity of protocols that are nondeterministic, or randomized, or Arthur-Merlin, and also to the information complexity of interactive protocols.

\end{abstract}

\section{Introduction}
Let us consider three strings $x, y, t$ and their Kolmogorov complexities $C(x), C(y)$, and respectively, $C(t)$.  It is sometimes useful to use the Venn diagram in Figure~\ref{f:figone} to visualize the information relations between the three strings.

 \begin{figure}[h]
 \centering{
\begin{tikzpicture}[shorten >=1pt,scale=0.12,
MyNode/.style={fill=green, circle},
]

\coordinate (x) at (0,0);
\coordinate (a) at (-7,0);
\coordinate (y) at (14,0);
\coordinate (b) at (19,0);
\coordinate (z) at (6,-14);
\coordinate (c) at (6,-21);
\coordinate (c0) at (6, 0.2);
\coordinate (c4) at (7, -3.8);

\node(c1) at (x) [draw, thick, fill=gray, circle through = (y), opacity=0.5] {};
\node(c2) at (y) [draw, thick, fill=gray, circle through = (x), opacity=0.5] {};
\node(c3) at (z) [draw, thick, fill=gray, circle through = (c0), opacity=0.5] {};

\node(X) at (a) {\footnotesize{$C(x)$}};
\node(Y) at (b) {\footnotesize{\quad$C(y)$}};
\node(Z) at (c)  {\footnotesize{$C(t)$}};
\node(I) at (c4)  {\textbf{\footnotesize{$I(x:y:t)$}}};

\end{tikzpicture}
}
\label{f:figone}
\caption{Three strings and their joint information relation}
\end{figure}
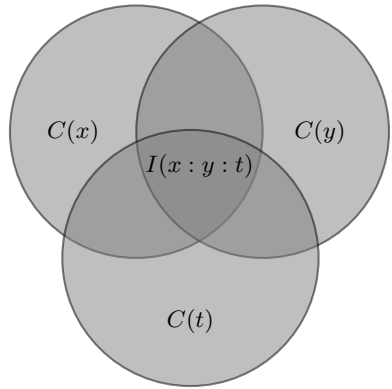
For example, the region contained in the left circle which lies outside the right circle can be thought to represent $C(x \mid y)$, the region at the intersection of the left and right circles can be thought to represent the mutual information $I(x:y)$, and so on.  There is however a nuisance: this visual representation is not always correct,  the potential trouble maker being the darker region at the intersection of the three circles. This region is denoted $I(x:y:t)$ and can be defined as $I(x:y) - I(x:y\mid t)$ (there are also some alternative definitions, which are equivalent up to an additive $O(\log n)$ term; see Lemma~\ref{l:formulas}). The problem is that $I(x:y:t)$ can be negative. Romashchenko and Zimand~\cite{rom-zim:c:mutualinfo} have shown that if $t$ is a computable function of $x$ and $y$ and if furthermore this function has the ``rectangle property"  stating that $t(x_1, y_1) = t(x_2, y_2) = t$ implies $t(x_1, y_2) = t(x_2, y_1) = t$, then actually $I(x:y:t)$ \emph{is} positive up to $O(\log n)$ precision, where $n = \max(|x|,|y|)$.

\begin{theorem}[~\cite{rom-zim:c:mutualinfo}]
\label{t:ineq}
For every computable $t$ with the rectangle property, for every pair of $n$-bit strings $(x,y)$,  
\begin{equation}
\label{e:rz}
I(x:y:t) \geqp 0,
\end{equation}
where $t$ denotes $t(x,y)$, and $\geqp$ hides a loss of precision bounded by $O(\log n)$.
\end{theorem}
\smallskip

A related result has been obtained by Kaced, Romashchenko and Vereshchagin~\cite{ka-rom-ver:j:conditional}.  It uses Shannon entropy instead of Kolmogorov complexity, and it was an inspiration for~\cite{rom-zim:c:mutualinfo}.

The inequality in Theorem~\ref{t:ineq} is particularly interesting in communication complexity because in this theory the rectangle property plays a prominent role.
 There are various models in communication complexity, the most basic one being the two-party model introduced in 1979 by Yao~\cite{yao:c:comm}. Alice and Bob want to compute a function $f(x,y)$ of two arguments, but Alice receives only $x$, and Bob receives only $y$.  To achieve their goal, they run an interactive protocol $\Pi$ (i.e., they exchange messages in several rounds, where each party computes the current message to be sent from his/her input and the previous messages) which allows them at the end to compute $f(x,y)$.  The string which encodes in some canonical way the sequence of messages exchanged by Alice and Bob on input strings $(x,y)$ is denoted $t(x,y)$, and is called the \emph{transcript} of the protocol.  The key observation is that if on  input pairs $(x_1, y_1)$ and $(x_2, y_2)$ the transcript of the protocol is the same string $t$, then the transcript on input $(x_1, y_2)$ will be $t$ as well, and therefore, the transcript function has the rectangle property.

  The above observation (which is standard in communication complexity, see~\cite{kus-nis:b:commcomplexity}) shows that a protocol induces a partition of the domain of inputs into rectangles, where each part (i.e., rectangle)  $R$ of the partition corresponds to a fixed value $t$ of the transcript, via the definition $R = \{(x,y) \mid t(x,y) = t \}$. Suppose now that the protocol  allows Alice and Bob to compute $f(x,y)$. Then we can think that the domain of $f$ is formed by the \emph{cells} $(x,y)$ and each cell is colored with the color $f(x,y)$. The rectangles of the partition induced by the protocol are  \emph{$f$-monochromatic}, because all the cells $(x,y)$ that have the same transcript $t(x,y) = t$ have the same color $f$.

Thus, every \emph{deterministic} interactive protocol induces a partition of the domain into monochormatic rectangles. But not every monochromatic partition into rectangles corresponds to an interactive deterministc protocol and, moreover,
in some applications, one analyzes  \emph{coverings} of the domain of $f$ with monochromatic rectangles that can overlap (so the rectangles are not necessarily a partition of the domain). Such coverings have extremely interesting applications (for example, the breakthrough result of Fiorini et. al.~\cite{fio-mas-pok-tiw-wol:j:polytope} uses them), and  characterize \emph{nondeterministic} communication complexity.
Given a function $f$ of two arguments, it is of interest to determine the minimum number of $f$-monochromatic rectangles that cover the domain of $f$. The logarithm of this number is the nondeterministic communication complexity of $f$.  (We note that the standard definition of nondeterministic communication complexity is for boolean functions -- see~\cite{kus-nis:b:commcomplexity}, Def 2.3 -- and, for this class of functions, the definition is in terms of coverings with $1$-monochromatic rectangles. In this work   we focus mainly on non-boolean functions, for which  the above definition is appropriate).

This is a ``worst-case" type of complexity, but using the framework of Kolmogorov complexity adopted in this work, we can talk about  the communication complexity for each individual  input. A nondeterministic communication protocol $\Pi$ can be described combinatorially,  but also as an interactive computational  procedure. In the combinatorial view, $\Pi$ is simply a covering of $f$ with $f$-monochromatic rectangles. In the procedural view, besides Alice and Bob, there is a third party, the Prover, also known as Merlin. Merlin knows both $x$ and $y$ (where $x$ is the input possessed by Alice, and $y$ is the input possessed by Bob).   Merlin sends Alice and Bob  a description of $t(x,y)$, which is one of the rectangles in the covering specified by  $\Pi$ that contains the cell $(x,y)$. Alice checks that $x$ is a row of $t(x,y)$ and Bob checks that $y$ is a column of $t(x,y)$, and, if both parties confirm that both checks were valid (which requires only two bits of communication), Alice and Bob  derive the coveted $f(x,y)$ which is just the color of $t(x,y)$.  

Thus, it is natural to define the \emph{individual} communication complexity of the protocol $\Pi$ on input $(x,y)$ to be $C(t(x,y)  \mid \Pi)$, \ie, the length of the shortest description of $t(x,y)$, given the rectangles of the protocol.
\medskip

\textbf{Our contributions.}  The center piece of this paper is an extension of Theorem~\ref{t:ineq}, which is applicable to rectangle covers.  The setting is as follows: ${\cal R}$ is a set of rectangles, and $\rho$, called the \emph{thickness} of ${\cal R}$, is the maximum number of rectangles in ${\cal R}$ that overlap. A pair $(x,y)$ is covered by ${\cal R}$ if  $(x,y)$ belongs to some rectangle of ${\cal R}$.
\begin{theorem}[Main Inequality, informal and simplified form: the full version is in Theorem~\ref{t:extension}]
\label{t:mainineq}
If $(x,y)$ is covered by ${\cal R}$ and $t$ is a rectangle of ${\cal R}$ containing $(x,y)$, then 
\begin{equation}
\label{e:mainineq}
I(x:y:t \mid {\cal R}) + \log \rho \geqp 0.
\end{equation}
\end{theorem}
The proof of Theorem~\ref{t:mainineq} is essentially the same as the proof of Theorem~\ref{t:ineq}, but we believe that the new inequality~\eqref{e:mainineq} deserves attention because it is applicable to the communication complexity of nondeterministic protocols, and also of Arthur-Merlin (AM) protocols, which combine nondeterminism with randomness. This is intriguing because currently there is a lack of techniques  for proving communication complexity lower bounds for AM protocols, and, consequently, finding AM-complexity lower bounds for any explicit function is a notoriously challenging open question in communication complexity~\cite{cha-tha:c:amcomm,goo:-pit-wat:j:unambig-am,goo-pit-wat:j:landscape,gav:t:am-communication}.  

Because of the $\log \rho$ term, the inequality~\eqref{e:mainineq} is meaningful only for nondeterministic and AM protocols with rectangle coverings having bounded thickness.  Of course this is a limitation, but communication protocols with small thickness have their merits and have been studied since long with various monikers for the thickness parameter  such as ``few witnesses" or  ``limited number of advice bits" \cite{kar-new-sak-wig:j:nondetcomm,hromkovic1996nondeterministic,hromkovivc2000tradeoffs,gro-tar:j:nondetcomm}. Yannakakis~\cite{yan:j:opt-communication} shows that communication-efficient nondeterministic protocols with thickness $1$ (these are called \emph{unambiguos nondeterministic protocols})  can be used to express certain combinatorial optimization problems as linear programs with few variables.
Furthermore, recent works of G\"{o}\"{o}s et. al.~\cite{goo:-pit-wat:j:unambig-am} and Gavinsky~\cite{gav:t:am-communication} highlight the importance of the thickness parameter in studying the complexity of AM communication protocols and analyze several models of AM-like protocols with thickness bounded in various ways. 
We believe that the information-theoretic inequality~\eqref{e:mainineq} can contribute to this research line and, more generally, to the understanding of AM protocols.  In Section~\ref{s:am}, we derive almost directly from a variant of~\eqref{e:mainineq}, a lower bound for the communication cost of AM protocols, albeit for the easier case of a non-boolean function. We show that any AM protocol that computes $x+y$ (bitwise XOR) must have communication cost $2n - \log(\rho)$, where $\rho$ is the thickness of the protocol.  

In Section~\ref{s:multiparty}, we present the analog of inequality~\eqref{e:mainineq} for multiparty communication protocols and give an application.

The \emph{information complexity} of a $2$-party interactive protocol (\cite{cha-shi-wir-yao:c:infcompl,bar-bra-che-rao:c:infcompl}) measures the amount of information each party learns about the other party's input during the execution of the protocol. Information complexity has turned out to be a very useful concept in communication complexity (for example, see the survey paper~\cite{wei:j:infcomplexity}).  Intuitively, the information complexity should not be larger than the communication complexity, because each bit of the transcript carries at most one bit of information. This relation has been proven to hold by Braverman and Rao~\cite{bra-rao:c:infcomplexity} in the case of \emph{randomized} protocols, but it is natural to ponder about the relation between information complexity and communication complexity for \emph{nondeterministic} protocols, where the  intuitive view is less clear. In Section~\ref{s:infcomplexity}, we consider the Kolmogorov complexity version of \emph{information complexity}.  Relying again on the  information-theoretical inequality~\eqref{e:mainineq}, we  show that in the case of a nondeterministic protocol $\Pi$, the information complexity of $\Pi$ is at most the communication complexity of $\Pi$ plus the logarithm of the thickness of $\Pi$, up to logarithmic precision.

\section{Prerequisites, notation, and some useful lemmas}
\label{s:basics}
We assume familiarity with the basics of Kolmogorov complexity theory. We use standard notation in Kolmogorov complexity.  We use $x, y, u, p,$ etc.  to denote finite binary strings, $|x|$ denotes the length of  string $x$. We fix a universal Turing machine $U$, and we say that $p$ is a program (or description) of $x$ conditioned on $u$ if $U$ on input $(p,u)$ prints $x$.  The Kolmogorov complexity of $x$ conditional on $u$ is $C(x \mid u) = \min\{|p| \mid p \mbox{ is a program for $x$ given $u$} \}$.  If $u$ is the empty string, we write $C(x)$ instead of $C(x \mid u)$.  

We use $\eqp, \geqp, \leqp$ to denote $=, \geq, \leq$ with a loss of precision bounded by $O(\log n)$, where  the constant hidden in the $O(\cdot)$ notation depends only on the universal Turing machine used in the definition of Kolmogorov complexity. The parameter $n$ is defined in the context, and, by default, is the maximum length of the strings involved in the relation.

Throughout the paper we use the notation $\ttt \stackrel{def.}{=} t(x,y)$ and $\fff \stackrel{def.}{=} f(x,y)$. 

The Kolmogorov-Levin theorem shows the validity of the chain rule $C(x,y) \eqp C(x) + C(y \mid x)$ (here $n$ hidden in the $\eqp$ notation, is $\max(|x|,|y|)$, as explained above).
 
 The mutual information of $x$ and $y$ conditioned on $u$ is denoted $I(x : y \mid u)$ and is defined by
$I(x : y \mid u) = C(x \mid u) + C(y \mid u) - C(x,y \mid u)$. In case $u$ is the empty string, we simply write $I(x:y)$. For every strings $x,y,z,u$, it holds that $I(x,z:y \mid u) \eqp I(z:y \mid u) + I(x:y\mid z,u)$ (the chain rule for mutual information).

\begin{lemma}
\label{l:formulas}
Let $I(u:v:w) = C(u) + C(v) + C(w) - C(u,v) - C(u,w) - C(v,w) + C(u,v,w)$.
Then
\begin{itemize} 
\item[(1)] $I(u:v:w)  \eqp I(u:v) - I(u:v \mid w)$.
\item[(2)]  $I(u:v:w) \eqp C(w) - C(w \mid u) - C(w \mid v)  + C(w \mid u,v)$.
\end{itemize}
\end{lemma}
\begin{proof}
Simple manipulations using the chain rule.
\end{proof}

\subsection{Nondeterministic protocols}
A rectangle $R$ is a set (usually a subset in $\{0,1\}^{n_1}\times \{0,1\}^{n_2}$) of the form $S \times T$.  We say that the set $S$ is the set of \emph{rows} of $R$, and $T$ is the set of \emph{columns} of $R$. 

 We  define  a nondeterministic communication protocol that allows two parties to compute a function $f(x,y)$, as being a covering $\rt$ with rectangles of the domain of $f$ together with a function $t(x,y)$ that selects one of  the rectangles containing $(x,y)$. The formal definition is as follows.
\begin{definition}
\label{d:protocol} \hfill
\begin{enumerate}
\item A nondeterministic communication protocol $\Pi$  for a function with two arguments is a pair $(\rt, t)$, where $\rt$ and $t$ are as follows.

\item The domain of the protocol is a set $D$ of the form  $D = \zo^{n_1} \times \zo^{n_2}$, for some  positive integers $n_1$ and $n_2$ and we let $n$ denote $\max(n_1, n_2)$. We view $D$ as a table with $2^{n_1}$ rows and $2^{n_2}$ columns.

\item $\rt$ is  a covering of $D$  with rectangles. That is  $\rt = \{R_1, \ldots, R_T\}$, for some natural number $T$, where each $R_i$ is a rectangle, and
$D \subseteq \bigcup R_i$.


\item $t$ is a function which on input $(x,y)$  returns one of the rectangles in $\rt$ that contains the cell $(x,y)$.  We think of $t(x,y)$ as being the transcript of the protocol on input $(x,y)$.

\item The communication complexity of the protocol $\Pi =(\rt, t)$ on input $(x,y)$ is $C(t(x,y)  \mid \Pi)$.
\end{enumerate}
\end{definition}
Henceforth, by ``protocol" we mean a nondeterministic communication protocol, unless specified otherwise. In the next definition we define the \emph{thickness} of a protocol, a parameter which quantifies how far from a partition is the covering induced by the protocol.
\begin{definition}
\label{d:thick}
Let $\Pi = (\rt, t)$ be a protocol. The \emph{thickness of a cell $(x,y)$}, denoted $\rho(\Pi_{x,y})$ is the number of rectangles in $\rt$ that contain the cell $(x,y)$.   
When the protocol $\Pi$ is clear from the context, we write more simply $\rho_{x,y}$. 
The \emph{thickness of the protocol}, denoted $\rho(\Pi)$,  is the maximum of the thickness of all cells in the domain $D$, \ie,  $\rho(\Pi) = \max_{(x,y) \in D} \rho(\Pi_{x,y})$.
The thickness of a rectangle $R$, denoted $\rho(R)$,  is $\max_{(x,y) \in R} \rho(\Pi_{x,y})$.
\end{definition}
We next define what it means for a protocol to compute a function exactly (\ie, on all inputs).
\begin{definition}
\label{d:protocolcompute} 
 A protocol $\Pi = (\rt, t)$ computes the function $f$ of two arguments over domain $D$, if every rectangle in $\rt$ is $f$-monochromatic, \ie, for every rectangle $R \in \rt$, there is $z$ (the ``color" of $R$) such that for all $(x,y) \in R$, $f(x,y) = z$.
\end{definition}

We also define what it means for a protocol to compute a function with some error (\ie, it may err on a small fraction of inputs). This case is useful in the definition of Arthur-Merlin protocols. In case the protocol makes mistakes on some inputs, we can no longer assume that the rectangles are monochromatic, and Alice and Bob use functions $g_A$ (respectively $g_B$) to compute the output from their input and the rectangle provided by Merlin.
\begin{definition}
\label{d:protocol-error} 
A protocol that computes the function $f$ with error  has the form $\Pi=(\rt, t, g_A, g_B)$  where
\begin{enumerate}
\item $\rt$ and $t$ are as in Definition~\ref{d:protocol}.
\item $g_A$ and $g_B$ are functions that map an input of Alice (respectively, an input of Bob) and a rectangle in $\rt$ into a element in the range of $f$.
\end{enumerate}
The protocol $\Pi$ computes $f$ with error $\epsilon$, if with probability $1-\epsilon$, $g_A(x,t(x,y)) = g_B(y, t(x,y)) = f(x,y)$, where the probability is over $(x,y)$ chosen uniformly at random in the domain $D$.
\end{definition}

\begin{lemma}
\label{l:mutualf}
If $\Pi= (\rt, t)$ is a protocol  that computes the computable function $f$, then for every $(x,y)$:
\begin{itemize}
\item[(1)] $I(x:y:\fff \mid \Pi) \eqp C(\fff \mid \Pi) - C(\fff \mid x, \Pi) - C(\fff \mid y, \Pi)$.
\item[(2)]  $I(x:y:\ttt \mid \Pi) \eqp C(\ttt \mid \Pi) - C(\ttt \mid x, \Pi) - C(\ttt \mid y, \Pi)$.
\end{itemize}
\end{lemma}
\begin{proof}
Both statements follow from Lemma~\ref{l:formulas} (2)  (relativized with $\Pi)$ taking into account that $C(\fff \mid x, y, \Pi) = O(1)$ and $C(\ttt \mid x, y, \Pi) = O(1)$.
\end{proof}


The above definitions can be extended for the case of protocols that compute relations. The idea is that Alice having $x$ and Bob having $y$ want to compute some value $z$ such that $(x,y,z)$ is in some given relation.

Let us consider a relation $\Gamma \subseteq D_1 \times D_2 \times D_3$.  We assume that for every $(x,y) \in D_1 \times D_2$ there is some $z \in D_3$ such that $(x,y,z) \in \Gamma$. A $\Gamma$-monochromatic rectangle is a rectangle $R \subseteq D_1 \times D_2$ such that there exists $z$ with the property that $(x,y,z) \in \Gamma$ for all cells $(x,y) \in R$. The smallest such $z$ (in some predefined linear order of $D_3$) is the \emph{color} of $R$. 
\begin{definition}
\label{d:protocolrelation} \hfill
\begin{enumerate}
\item A protocol $\Pi= (\rt, t)$ computes the relation $\Gamma$, if every rectangle in $\rt$ is $\Gamma$-monochromatic.

\item A protocol $\Pi= (\rt, t, g_A, g_B)$  (see  Definition~\ref{d:protocolcompute}(2)) computes the relation $\Gamma$ with error $\epsilon$, if with probability $1-\epsilon$,  $g_A(x, t(x,y)) = g_B(y, t(x,y)) = \stackrel{def.}{z}$ and $(x,y, z) \in \Gamma$, where the probability is over $(x,y)$ chosen uniformly at random in the domain $D$.
\end{enumerate}
\end{definition}

\section{The main inequality}
\label{s:main}

The following result is an extension of Theorem~\ref{t:ineq}. Its proof follows closely the proof of Theorem~\ref{t:ineq} 
from~\cite{rom-zim:c:mutualinfo}.

\begin{theorem}
\label{t:extension}
For every protocol $\Pi=(\rt, t)$, and for every $(x,y) \in D$, 
\begin{equation}
\label{e:eqmain}
I(x:y \mid \Pi) \geqp I(x:y \mid \ttt, \Pi) - \log \rho(\ttt).
\end{equation}
\end{theorem}
\begin{remark}
\label{r:remproof}
Similar-looking inequalities have been used implicitly in other papers studying interactive protocols (for example,~\cite[Lemma 2.2]{ahl-csi:j:seckeyone},~\cite[Lemma 3.14]{bra-rao:c:infcomplexity}).  Their proofs have an inductive structure based on the rounds of communication. For instance, consider a $2$-round protocol where Alice sends message $t_1$ in Round $1$, and Bob sends message $t_2$ in Round $2$. We can show that $I(x:y) \geq I(x:y \mid t_1, t_2)$ as follows: 
\[
\begin{array}{lll}
I(x:y) & = I(x,t_1:y) & \mbox{(because $t_1$ is determined by $x$)} \\
 & \geq I(x : y \mid t_1) & \mbox{(chain rule and dropping a positive term)} \\
& =I(x : y,t_2 \mid t_1) & \mbox{(because $t_2$ is determined by $y$ and $t_1$)}\\
 &\geq I(x : y \mid t_1, t_2) &\mbox{(chain rule and dropping a positive term)}.
\end{array}
\]
This approach does not work for nondeterministic protocols, where Merlin's contribution to the transcript depends on both $x$ and $y$ and is delivered in ``one-shot,"  and not round-by-round. This is why we use a different method, using an idea from~\cite{rom-zim:c:mutualinfo}, where it was employed for a different reason.
\end{remark}
\begin{proof}
Let us fix $(x,y) \in D = \zo^{n_1} \times \zo^{n_2}$. We say that $x'$ is a clone of $x$ conditional on $\ttt$ if (i)   $x'$ is a row of the rectangle $\ttt$ and (ii) $C(x' \mid \Pi) \leq C(x \mid \Pi)$.  Similarly, we  say that $y'$ is a clone of $y$ conditional on $\ttt$ if (i) $y'$ is a column of the rectangle $\ttt$ and (ii) $C(y' \mid \Pi) \leq C(y \mid \Pi)$.  Let $\mathrm{Clones}_x$ denote the set of clones of $x$ (conditional on $\ttt$) and $\mathrm{Clones}_y$ denote the set of clones of $y$ (conditional on $\ttt$). Let $n$ be the maximum between the length of $x$ and the length of $y$.
\begin{claim}
$|\mathrm{Clones}_x| \geq 2^{C(x \mid \ttt, \Pi) - O(\log n)}$ and $|\mathrm{Clones}_y| \geq 2^{C(y \mid \ttt, \Pi) - O(\log n)}$.
\end{claim}
 \begin{proof} (of claim)
 Given $\ttt$, $\Pi$ and $C(x \mid \Pi)$ (with the observation that the latter can be written in a self-delimited way on $2 \log n$ bits), we can enumerate the clones of $x$. Since $x \in \mathrm{Clones}_x$, it follows that $C(x \mid \ttt, \Pi) \leq \log |\mathrm{Clones}_x| + 2 \log n$. Therefore
 \[
 |\mathrm{Clones}_x| \geq 2^{C(x \mid \ttt, \Pi) - 2\log n}.
 \]
 The other relation follows in the same way.
 \end{proof}
 Now let us take $(x',y')$ in $\mathrm{Clones_x} \times \mathrm{Clones_y}$ which maximizes $C(x',y' \mid \ttt, \Pi)$. Then
 \begin{equation}
 \label{e:eqt1}
 \begin{array}{ll}
 C(x',y' \mid \ttt, \Pi) &\geq \log |\mathrm{Clones_x} \times \mathrm{Clones_y}| \\
 &\geq C(x \mid \ttt, \Pi) + C(y \mid \ttt, \Pi) - O(\log n).
 \end{array}
 \end{equation}
 Next
 \begin{equation}
 \label{e:eqt2}
 \begin{array}{ll}
 C(x',y',\ttt  \mid \Pi) &\geqp C(\ttt \mid \Pi) + C(x',y' \mid \ttt, \Pi) \\
 &\geqp C(\ttt \mid \Pi) +  C(x \mid \ttt, \Pi) + C(y \mid \ttt, \Pi). \quad\quad\mbox{ (by~\eqref{e:eqt1})}
 \end{array}
 \end{equation}
 Also
  \begin{equation}
 \label{e:eqt3}
 \begin{array}{ll}
 C(x \mid \Pi)  + C(y \mid \Pi) &\geqp C(x' \mid \Pi) + C(y' \mid  \Pi)  \quad\quad \mbox{ (by def. of clones)}\\
 &\geqp C(x' ,  y' \mid \Pi)  \\
 &\geqp C(x',y', \ttt \mid \Pi) - \log \rho(\ttt).
 \end{array}
 \end{equation}
 The reason for the last inequality is that $(x',y')$ belongs to at most 
 $\rho_{x',y'}$ rectangles of the covering $\Pi$. On the other hand, $(x',y')$ belongs to the rectangle $\ttt$ and thus $\rho_{x',y'} \leq \rho(\ttt)$. This implies that $C(\ttt \mid x',y', \Pi) \leq  \log  \rho (\ttt)+ O(1)$.   Combining inequalities~\eqref{e:eqt2} and~\eqref{e:eqt3},
  \[
  C(\ttt \mid \Pi) + C(x \mid \ttt, \Pi) + C(y \mid \ttt, \Pi) \leqp C(x \mid \Pi) + C(y \mid \Pi) +   \log \rho (\ttt).
  \]
 Subtracting $C(x,y, \ttt \mid \Pi)$ in the left hand side and (the smaller) $C(x,y \mid \Pi)$ in the right hand side, we obtain
 \[
 I(x:y \mid \ttt, \Pi) \leqp I(x:y \mid \Pi) +   \log \rho(\ttt),
 \]
 which concludes the proof.
 \end{proof}



\section{Lower bounds for the communication complexity of 2-party protocols}
\label{s:twoparty}

The crux in communication complexity is  proving lower bounds for concrete problems, because such lower bounds can be transferred  to other domains (data structures, streaming algorithms, circuit complexity, and many other ones).   Theorem~\ref{t:extension} has some consequences, which can be used to prove lower bounds, as we will show below.

Consider a $2$-party deterministic and computable protocol which allows Alice and Bob to compute a function $f$ on input $(x,y)$, when Alice has $x$ and Bob has $y$. The transcript $\ttt$ has two parts:  $t_A$, comprising the messages sent by Alice, and $t_B$, comprising the messages sent by Bob. Clearly,  $C(t_A) \geq C(\fff \mid y)$ because 
$\fff \stackrel{def.}{=} f(x,y)$ can be computed from $y$ and  $t_A$.  Similarly, $C(t_B) \geq C(\fff \mid x)$.  In this way we can compute lower bounds for $t_A$ and $t_B$ from bounds of the conditional complexity of $\fff$. We would like to do the same thing for $\ttt$.  In a deterministic protocol, we do have that $\ttt = (t_A, t_B)$, but we cannot say directly that $C(\ttt) \geq C(\fff \mid x) + C(\fff \mid y)$ because it is not clear if $C(\ttt) \geqp C(t_A) + C(t_B)$.  In a nondeterministic communication protocol, $\ttt$ is provided by Merlin and there are no $t_A$ and $t_B$. Nevertheless, we show that for every protocol  $\Pi = (\rt, t)$, $C(\ttt \mid \Pi) \geqp C(\fff \mid x, \Pi) + C(\fff \mid y, \Pi)- \log \rho(\Pi)$. In fact, Theorem~\ref{t:bound-1} shows that a stronger bound holds true.

\begin{theorem} 
\label{t:bound-1}
\hfill
For every computable function $f$, for every protocol $\Pi= (\rt, t)$ that computes $f$ over domain $D$, and for every $(x,y) \in D$,
\begin{equation}
\label{e:main}
C(\ttt \mid \Pi) \geqp C(\fff \mid x, \Pi) + C(\fff \mid y, \Pi) + C(\ttt \mid x, \fff, \Pi) + C(\ttt \mid y, \fff, \Pi) - \log \rho(\ttt).
\end{equation}
\end{theorem}

\if01
\begin{remark}
\label{r:thickness}
Equation~\eqref{e:main} is not very useful if protocols have large thickness.  In our examples in the following sections, the results are interesting for protocols with thickness
$\rho(\Pi)$ bounded by $2^{o(n)}$. Protocols with restricted thickness have been studied in the literature under the  moniker ``nondeterministic protocols with restricted advice," see~\cite[Section 4.7]{juk:b:bool}.   We note that if a protocol has thickness larger than quasi-polynomial, then the communication cost is larger than polylogarithmic  and the  protocol is inefficient from the point of view of communication complexity.
\end{remark}
\fi
Theorem~\ref{t:bound-1} is an immediate corollary of the following lemma which relaxes the condition that rectangles are $f$-monochromatic by requiring only that $\fff$ can be computed from the rectangle $\ttt$ and $x$ and also from the rectangle $\ttt$ and $y$ (and we can even allow $O(\log n)$ help bits in the computation).

\begin{lemma} 
\label{l:bound}
Let $\Pi = (\rt, t)$ be a protocol over domain $D$. For $(x,y) \in D$, let $\fff$ be a string such that $C(\fff \mid \Pi, \ttt, x) = O(\log n)$ and 
$C(\fff \mid \Pi, \ttt, y) = O(\log n)$.  Then,  
\begin{equation}
\label{e:main1}
C(\ttt \mid \Pi) \geqp C(\fff \mid x, \Pi) + C(\fff \mid y, \Pi) + C(\ttt \mid x,\fff, \Pi) + C(\ttt \mid y, \fff, \Pi) - \log \rho(\ttt).
\end{equation}

\end{lemma}
\begin{proof} 
To simplify the notation, we drop in all the $C(\cdots)$, $I(\cdots)$ terms below the conditioning on $\Pi$. For example, with this notational convention, the conclusion becomes 
\[
C(\ttt) \geqp C(\fff \mid x) + C(\fff \mid y ) + C(\ttt \mid x,\fff) + C(\ttt \mid y, \fff) - \log \rho(\ttt).
\]

From Theorem~\ref{t:extension}, 
\[
I(x:y) \geqp I(x:y \mid \ttt) -  \log \rho(\ttt).
\]
 This can be written as 
 \[
 C(x) + C(y) - C(x,y) \geqp C(x \mid \ttt) + C(y \mid \ttt) - C(x,y \mid \ttt) - \log \rho(\ttt),
 \]
  which taking into account that  $C(x,y, \ttt) \leq C(x,y) +  O(1)$ 
(because, given $\Pi$,  $\ttt$ can be computed from $x,y$)  implies 
\begin{equation}
\label{e:eq1}
C(x) + C(y) \geqp C (x \mid \ttt) + C(y \mid \ttt) + C(\ttt) - \log \rho(\ttt).
\end{equation} 
In the right hand side,
\[
\begin{array}{ll}
C(x \mid \ttt) & \eqp C(x,\ttt) - C(\ttt) \\
& \eqp C(x,\ttt,\fff) - C(\ttt) \\
& \eqp C(x, \fff) + C(\ttt \mid x, \fff) - C(\ttt).
\end{array}
\]
We have used the fact that $C(x, \ttt) \eqp  C(x,\ttt, \fff)$ which follows from the lemma's  hypothesis regarding the complexity of $\fff$.
Similarly,
\[
C(y \mid \ttt) \eqp C(x, \fff) + C(\ttt \mid x, \fff) - C(\ttt).
\]
Plugging these inequalities in Equation~\eqref{e:eq1}, 
\[
C(x) + C(y)  \geqp C(x,\fff) + C(\ttt \mid x, \fff) + C(y,\fff) + C(\ttt \mid y, \fff) - C(\ttt) - \log \rho(\ttt),
\]
which can be rewritten as
\[
C(\ttt) \geqp C(\fff \mid x) + C(\fff \mid y) +C(\ttt \mid x,\fff) + C(\ttt \mid y,\fff) - \log \rho(\ttt).
\]
\if01
\begin{equation}
\label{e1}
\begin{array}{ll}
C(a \mid t) & \eqp C(a,t)  - C(t)\\
& \eqp C(a,t, f(a,b))  - C(t)\\
& \eqp C(a, f(a,b)) + C(t \mid   a,f(a,b)) - C(t) \\
        & \eqp C(a \mid f(a,b)) +  C(t \mid   a,f(a,b))  + C(f(a,b)) - C(t).
\end{array}
\end{equation}
Similarly (by relativizing with $b$),
\begin{equation}
\label{e2}
\begin{array}{ll}
C(a \mid b, t) & \eqp   C(a \mid b, f(a,b)) +  C(t \mid  b,  a,f(a,b))  + C(f(a,b) \mid b) - C(t \mid b) \\
& \eqp C(a \mid b, f(a,b))  + C(f(a,b) \mid b) - C(t \mid b).
\end{array}
\end{equation}
Subtracting \eqref{e2} from \eqref{e1}, we obtain
\begin{equation}
\label{e3}
\begin{array}{ll}
I(a:b \mid t) & \eqp   I (a : b \mid f(a,b)) +  C(t \mid  a,f(a,b))  + C(f(a,b)) - C(t) -  \\
& \hspace{6cm}  - C(f(a,b) \mid b) +  C(t \mid b),
\end{array}
\end{equation}
which implies
\begin{equation}
\label{e4}
\begin{array}{ll}
C( t) & \eqp   I (a : b \mid f(a,b))  - I(a:b \mid t) +  C(t \mid  a,f(a,b))  + C(f(a,b)) -   \\
& \hspace{6cm}  - C(f(a,b) \mid b) +  C(t \mid b) \\
& \geqp  I (a : b \mid f(a,b))  - I(a:b)  + I(a:b:t) +  \\
& \hspace{5cm} C(t \mid  a,f(a,b))  + C(f(a,b)) -  C(f(a,b) \mid b) +  C(t \mid b) \\
&\hspace{5cm}\mbox{(by Lemma~\ref{l:andrei})}.
\end{array}
\end{equation}
Next, we use Lemma~\ref{l:mutualf} and the fact that $C(t \mid b) \geq C(f(a,b) \mid b)$ (because, $f(a,b)$ can be computed from $b$ and $t$), and we obtain
\[
C(t) \geqp  C(f(a,b) \mid a) + C(f(a,b) \mid b) + I(a:b:t) +  C(t \mid a, f(a,b)) ,
\]
which is equivalent to the conclusion of the theorem.
\fi
\end{proof}
 

Similar bounds hold true for other mechanisms by which a communication protocol performs a computational task, i.e., for protocols computing functions with small error, or protocols that compute relations with and without error. 
\begin{theorem} 
\label{t:bound-2}
\hfill
\begin{enumerate}
\item In case the protocol $\Pi = (\rt, t, g_A, g_B)$ computes $f$ over domain $D$ with error $\epsilon$, then the inequality~\eqref{e:main} holds with probability $1- O(\epsilon)$ over $(x,y)$ chosen uniformly at random in $D$.
\item Suppose the protocol $\Pi = (\rt, t)$ computes the relation $\Gamma \subseteq D_1 \times D_2 \times D_3$ and let $\fff$ be the color of rectangle $\ttt$.  Then the inequality~\eqref{e:main} holds true for every $(x,y)$ in $D_1 \times D_2$.
\item  Suppose the protocol $\Pi = (\rt, t)$ computes the relation $\Gamma \subseteq D_1 \times D_2 \times D_3$ with error $\epsilon$ and let $\fff\stackrel{def.}{=} g_A(x, \ttt) = g_B(y, \ttt)$ (in case the latter two values are not equal, then $\fff$ is not defined). Then the inequality~\eqref{e:main} holds true with probability $1-O(\epsilon)$ over $(x,y)$ chosen uniformly at random  in $D_1 \times D_2$.
\end{enumerate}
\end{theorem}

\section{Application to Arthur-Merlin protocols}
\label{s:am}
Equation~\eqref{e:main}  can be used to establish lower bounds for protocols that use randomness, or mix  nondeterminism and randomness. In the latter type of protocols, Merlin provides a proof (like in standard nondeterministic protocols), and Alice and Bob (which together play the role of Arthur) probabilistically verify the proof and next compute the common output.  There are actually two types of such protocols: Merlin-Arthur (MA)  protocols in which the randomness is shared between Alice and Bob but is not visible to Merlin, and Arthur-Merlin (AM) protocols in which the randomness is shared between all parties, Alice, Bob, and Merlin.   We provide here a lower bound for AM protocols, which is tight for protocols whose thickness is not too large. The lower bound is valid for all 2-party (\ie, without Merlin) randomized protocols. To our knowledge, this result  does not seem to be attainable by other methods.  We note that lower bounds for AM protocols are considered to be difficult, while, equipped with the results from the previous section, our proof is short and easy.  For a discussion on AM protocols, the reader can see the survey paper~\cite{goo-pit-wat:j:landscape}, where lower bounds for AM protocols are considered to be beyond reach at the current time; however, this consideration  refers to lower bounds for boolean functions, and our example concerns a non-boolean function.

An AM protocol is essentially a distribution over nondeterministic protocols. More precisely, Alice, Bob and Merlin share a source of randomness. For each $r$ drawn from the source, there is a protocol $\Pi_r = (\rt_r, t_r, g_{r,A}, g_{r,B})$ as in Definition~\ref{d:protocol-error}.  The protocol computes the function $f$ if for every $(x,y)$ in the domain of $f$, with probability of $r$ at least $2/3$,
\begin{equation}
\label{e:eq1.1}
g_{r,A}(x,t_r(x,y)) = g_{r,B}(y,t_r(x,y)) = f(x,y).
\end{equation}
The communication cost of an AM protocol is $\log  (\max_r |\rt_r|)$, \ie,   the logarithm of the maximum number of rectangles, over all randomness $r$. The thickness of an AM protocol $\Pi$ is denoted $\rho(\Pi)$ and is by definition the maximum thickness of all protocols $\Pi_{r}$.

We now present our first application. The arguments below are valid for any group, but for concreteness, let us consider the group $G = (\mathbb{Z}_2^n, +)$. Alice and Bob want to compute $f(x,y) = x+y$, where 
Alice has $x \in \mathbb{Z}_2^n$ and Bob has  $y \in \mathbb{Z}_2^n$.

A straightforward protocol consists in Merlin providing $x$ and $y$. Alice and Bob (without actually using randomness) check and confirm to each other that Merlin  has provided their inputs, after which they compute $f(x,y)$. This protocol has thickness $1$, and the communication cost is $2n+2$ (for $x$ and $y$ and for the two confirmation bits). There is an MA protocol with communication $n + O(1)$: Merlin sends a value $z$ claiming $z = x+y$, and then Alice and Bob  with $O(1)$ additional communication  and  using their shared randomness (which is secret to Merlin) can check that $z-x$ and $y$ are equal using random fingerprinting in the standard way.
This strategy does not work in an AM protocol, because if Merlin knows the randomness, he can cheat by using a wrong $z$ that passes the Alice/Bob test. Still,
 it is in principle conceivable that there may exist an AM protocol which is more communication-efficient than the trivial protocol above. We show that this is not possible for AM protocols whose thickness is not too large.

\begin{claim}
\label{c:sum}
Any AM protocol $\Pi$ that computes $f$  must have communication cost  at least $2n - \log (\rho(\Pi)) - O(1)$. Thus,  $2n$ is essentially a lower bound for AM protocols  with, say,  $2^{o(n)}$ thickness. 
\end{claim}

Let us consider an AM protocol that computes $f(x,y)$. For each randomness $r$, let ${\rm GOOD}_r$ denote the set of pairs $(x,y)$ which satisfy the relations in equation~\eqref{e:eq1.1}. Such a pair $(x,y)$  is said to be  correct with respect to $r$. 

Let us first attempt a direct argument which does not utilize the tools developed in the previous section.  By a standard averaging argument there is some $r_0$ for which ${\rm GOOD}_{r_0}$ contains at least $(2/3)$ of the input pairs $(x,y)$. No rectangle in $\rt_{r_0}$ can have more than $2^n$ pairs which are correct with respect to $r_0$. The reason for this is that the number of rows  in a rectangle is bounded by $2^n$ and each row  in the rectangle can have at most one correct pair (because two cells $(x,y_1)$ and $(x,y_2)$ in row $x$ are assigned the same value $z$ by $g_A$ and it is not possible that both $x+y_1$ and $x+y_2$ are equal to $z$).   Since the number of correct cells (\ie,   the size of   ${\rm GOOD}_{r_0}$) is at least $(2/3) 2^{2n}$, it follows that $\rt_{r_0}$ contains at least $(2/3)2^{2n}/2^n = 2^{n-O(1)}$ rectangles, and therefore the communication cost is at least $n-O(1)$, which is smaller than the claimed $2n$ lower bound.

We can do better by using Lemma~\eqref{l:bound}. 
Using this lemma and the same $r_0$ as in the argument above, we obtain that for every pair $(x,y)$, which is correct with respect to $r_0$ (and recall that there are at least $(2/3)2^{2n}$ such pairs),
\begin{equation}
\label{e:eq2}
C(t_{r_0}(x,y) \mid \Pi_{r_0} ) \geq C(x+y \mid x, \Pi_{r_0}) + C(x+y \mid y, \Pi_{r_0}) - \log (\rho(\Pi_{r_0})).
\end{equation}
 We next observe that for every $x$ and every string $z$, and for any $c$,
$C(x+y \mid x, z) \geq n- c$ for at least  a fraction of $(1-2^{-c})$ elements $y$  in $\mathbb{Z}_2^n$, because $x+y$ takes $2^n$ possible values and only $2^{n-c} - 1$ of them can have complexity less than $n-c$.  The similar relation holds if we swap $x$ and $y$.
Using these estimations   for the first two terms in the right hand side of Equation~\eqref{e:eq2} (with $\Pi_{r_0}$ in the role of $z$), we obtain
that for at least a fraction of $(2/3 - 2^{-c})$ of all pairs $(x,y)$, $C(t_{r_0}(x,y) \mid \Pi_{t_{r_0}}) \geq  2n - c - \log (\rho(\Pi_{r_0}))$. This implies that  the communication cost of any AM protocol is at least $2n - \log (\rho(\Pi)) - O(1)$, as claimed.

Regarding the claim about randomized protocols made at the beginning of this section, we note that a 2-party randomized protocol is a distribution over 2-party deterministic protocols, and thus they can be viewed as AM protocols with thickness equal to $1$. The above argument implies that any randomized protocol that computes $f(x,y) = x+y$ (in the group $\mathbb{Z}_2^n$) with probability $2/3$ has communication cost $2n - O(1)$  for the majority of input pairs $(x,y)$.

\if01
As we have remarked earlier (see Remark~\ref{r:thickness}), $\log (\rho(\Pi))$ is a lower bound of the communication cost. Then Claim~\ref{c:sum} implies that any AM protocol for 
$f(x,y) = x+y$ (bitwise XOR) has communication cost at least $n - O(1)$, which is the same as the lower bound obtained with the direct argument. There are situations where the direct argument does not seem to be applicable, and where we can use the results from the previous section. One such example is for $g(x,y)$ equal to the bitwise OR of $x$ and $y$ (\ie, Alice has the $n$-bit string $x= x_1 \ldots x_n$, Bob has the $n$-bit string $y= y_1 \ldots y_n$, and they want to compute the $n$-bit string $g(x,y) = (x_1 \vee y_1) \ldots (x_n \vee y_n)$).
\begin{claim}
\label{c:or}
Any AM protocol $\Pi$ that computes $g$  must have communication cost  at least $n/2 - 2\sqrt{n} - O(1)$.
\end{claim}
Let $\Pi$ be an AM protocol that computes $g$ with probability $2/3$ and let $r_0$ be as in the above argument. There are more than a fraction of $99/100$ of  $n$-bit strings $x$ having at least $n/2  - 2 \sqrt{n}$ $0$'s. For such a string $x$, when $y$ varies, $g(x, y)$ takes $2^{n/2- 2\sqrt{n}}$  different values. It follows that for some constant $c$, there are at least a $98/100$  fraction of pairs $(x,y)$ such that $C(g(x,y) \mid x,  \Pi_{r_0}) \geq n/2  - 2 \sqrt{n} - c$.  A similar fact holds if we swap $x$ and $y$.  Now we can use the analog of Equation~\eqref{e:eq2}, and obtain that for at least a fraction of $((2/3) - (96/100))$ of pairs $(x,y)$, 
\begin{equation}
\label{e:eq3}
C(t_{r_0}(x,y) \mid \Pi_{r_0} ) \geq 2(n/2 - 2 \sqrt{n} - c) - \log (\rho(\Pi_{r_0})).
\end{equation}
Either $\log(\rho(\Pi_{r_0}))$ is at least  $n/2 - 2\sqrt{n}$, or is less than $n/2 - 2\sqrt{n}$. In the former case, the communication cost of $\Pi$ is at least $n/2 - 2 \sqrt{n}$, because the number of rectangles is at least log of the thickness. In the second case, we use Equation~\eqref{e:eq3}, and we deduce again that the communication cost of $\Pi$ is at least $n/2 - 2 \sqrt{n} - O(1)$, and thus the claim is proved.

Regarding the claim about randomized protocols made at the beginning of this section, we note that a 2-party randomized protocol is a distribution over 2-party deterministic protocols, and thus they can be viewed as AM protocols with thickness equal to $1$. The above argument implies that any randomized protocol that computes $f(x,y) = x+y$ (in the group $\mathbb{Z}_2^n$) with probability $2/3$ has communication cost $2n - O(1)$  for the majority of input pairs $(x,y)$.
\fi

\if01
\begin{example}
\label{e:ex1}

Let $x$ and $y$ be two lines in the 2-dimensional affine plane over the field with $2^n$ elements and let $P$ be the point at the intersection of $x$ and $y$ (if $x$ and $y$ are parallel or equal, we let $P$ be, say, the origin).  Each one of  $x$ and $y$ can be represented by $2n$-bit strings (the slope and the intercept of each line written on $n$ bits), and $P$ also is represented with $2n$ bits for the two coordinates.   Suppose first that Alice receives $x$, Bob receives $y$, and they want to determine $P$. There is a straightforward nondeterministic communication protocol for this task  with communication cost $2n$ bits for every $(x,y)$: Merlin simply sends the point $P$ to Alice and Bob, who check that the point belongs to their lines $x$ and respectively $y$. (In case, $x$ and $y$ are parallel lines, Merlin can just send the common slope of the two lines, which takes $n$ bits).  This is clearly optimal in the \emph{worst-case} sense, because $P$ is $2n$ bits long, and thus there are $2^{2n}$ ``colors." Therefore the number of monochromatic rectangles has to be at least $2^{2n}$. 


We show that essentially no protocol with $\poly(n)$ thickness can do  better than the straightforward protocol in the \emph{average case}. We claim that for every $\epsilon > 0$, on a fraction of $(1-\epsilon)$ of the input pairs $(x,y)$, the transcript $t(x,y)$ of any protocol $\Pi$ with $\poly(n))$ thickness  has complexity $C(t(x,y) \mid \Pi) \geqp 2n - O(\log(1/\epsilon))$ (which means that Merlin needs to send  at least $2n - O(\log(1/\epsilon)) - O(\log n)$ bits for $(1-\epsilon)$ fraction of $(x,y)$).
 
 The claim is proven by the following argument. When $x$ is fixed and $y$ varies, the point $P$ at the intersection of lines $x$ and $y$, which we denote $f(x,y)$,  can be any any point on the $x$ line  and, therefore, $f(x,y)$ takes $2^n$ possible values. By a standard counting argument, for every $c$, for every $x$ and for $1-2^{-c}$ fraction of $y$'s, $C(f(x,y) \mid x, \Pi) \geq n -c$.	There is a similar bound for $C(f(x,y) \mid y, \Pi)$. It follows that for a $(1-2\cdot 2^{-c})$ fraction of $(x,y)$,  $C(f(x,y) \mid x, \Pi) + C(f(x,y) \mid y, \Pi) \geq 2n - c$.
 Now, we use Theorem~\ref{t:bound}   and obtain that for a $(1-2\cdot 2^{-c})$ fraction of $(x,y)$, $C(t(x,y) \mid \Pi) \geq 2n - c.$ The claim follows with $\epsilon = 2 \cdot 2^{-c}$.
 \if01
 One can show much more directly that any covering with $f$-monochromatic rectangles requires $2^{2n}$ rectangles, and hence the nondeterministic communication complexity of $f$  is $2n$. Indeed, since $f(x,y)$ is $2n$-bits long, there are $2^{2n}$ ``colors" and therefore the number of $f$-monochromatic rectangles has to be at least $2^{2n}$.  However, this simple argument only shows that the worst-case nondeterministic complexity of $f$ is $2n$, while the argument using Theorem~\ref{t:bound}  shows that almost the same lower bound holds for most inputs. Moreover, consider a modification of this example, in which Alice and Bob want to find only the first coordinate of the point at the intersection of their lines. The modified $f(x,y)$ takes $2^n$ values and thus the simple argument based on the number of colors gives a lower bound of $n$, while the above argument using Equation~\eqref{e:main} gives a lower bound of $2n - \log(1/\epsilon)$.

Let $x$ and $y$ be two lines in the 2-dimensional affine plane over the field with $2^n$ elements and let $f(x,y)$ be the point at the intersection of $x$ and $y$ (if $x$ and $y$ are parallel, we let $f(x,y) = 0$).  Note that $x$ and $y$ can be represented by $2n$-bit strings (the slope and the intercept of each line written on $n$ bits), and $f(x,y)$ also is represented with $2n$ bits for the two coordinates.  In the straightforward deterministic protocol, Bob sends $y$ to Alice, Alice computes the intersection point $f(x,y)$, and she sends $f(x,y)$ to Bob. The transcript is $4n$ bits long for every $(x,y)$. There is a nondeterministic communication protocol for computing $f(x,y)$ with communication cost $2n$ bits for every $(x,y)$: Merlin simply sends the point $f(x,y)$ to Alice and Bob, who check that the point belongs to their lines $x$ and respectively $y$. (In case, $x$ and $y$ are parallel lines, Merlin can just send the common slope of the two lines, which takes $n$ bits). 

We next show that essentially no protocol with $\poly(n)$ thickness can do  better. We claim that for every $\epsilon > 0$, on a fraction of $(1-\epsilon)$ of the input pairs $(x,y)$, the transcript $t(x,y)$ of any protocol $\Pi$ with $\poly(n))$ thickness  has complexity $C(t(x,y) \mid \Pi) \geqp 2n - O(\log(1/\epsilon))$ (which means that Merlin needs to send  at least $2n - O(\log(1/\epsilon)) - O(\log n)$ bits for $(1-\epsilon)$ fraction of $(x,y)$).
 
 The claim is proven by the following argument. When $x$ is fixed and $y$ varies, $f(x,y)$ can be any point on the line $x$, and, therefore, $f(x,y)$ takes $2^n$ possible values. By a standard counting argument, for every $c$, for every $x$ and for $1-2^{-c}$ fraction of $y$'s, $C(f(x,y) \mid x, \Pi) \geq n -c$.	There is a similar bound for $C(f(x,y) \mid y, \Pi)$. It follows that for a $(1-2\cdot 2^{-c})$ fraction of $(x,y)$,  $C(f(x,y) \mid x, \Pi) + C(f(x,y) \mid y, \Pi) \geq 2n - c$.
 Now, we use Theorem~\ref{t:bound} (1)  and obtain that for a $(1-2\cdot 2^{-c})$ fraction of $(x,y)$, $C(t(x,y) \mid \Pi) \geq 2n - c.$ The claim follows with $\epsilon = 2 \cdot 2^{-c}$.
 
 One can show much more directly that any covering with $f$-monochromatic rectangles requires $2^{2n}$ rectangles, and hence the nondeterministic communication complexity of $f$  is $2n$. Indeed, since $f(x,y)$ is $2n$-bits long, there are $2^{2n}$ ``colors" and therefore the number of $f$-monochromatic rectangles has to be at least $2^{2n}$.  However, this simple argument only shows that the worst-case nondeterministic complexity of $f$ is $2n$, while the argument using Theorem~\ref{t:bound}  shows that almost the same lower bound holds for most inputs. Moreover, consider a modification of this example, in which Alice and Bob want to find only the first coordinate of the point at the intersection of their lines. The modified $f(x,y)$ takes $2^n$ values and thus the simple argument based on the number of colors gives a lower bound of $n$, while the above argument using Equation~\eqref{e:main} gives a lower bound of $2n - \log(1/\epsilon)$.
 \fi
\end{example}
\medskip

\begin{example}
\label{e:ex2}
We show how to use Theorem~\ref{t:bound-2} (2) to obtain lower bounds for the communication complexity of randomized protocols. This is interesting because in general lower bounds for this type of protocols have difficult proofs.

We consider the same function $f(x,y)$ as in Example~\ref{e:ex1}, that is $f(x,y)$ is the point at the intersection of lines $x$ and $y$. Let us consider a probabilistic protocol that computes $f$ with error $\epsilon$ and communication complexity $k$. We are informal here, but this means that using random bits $r$ (which can be private or public), for every two lines $(x,y)$ in the given domain, Alice having line $x$  and Bob having line $y$, after exchanging messages, compute at the end of the protocol the intersection point $f(x,y)$ with probability $1-\epsilon$ over $r$ chosen uniformly at random,  and that the transcript is at most $k$ bits long on every $(x,y)$. 

We claim that $k \geq 2n - \log(1/\epsilon) - O(1)$.

By a standard averaging argument, there is $r_0$ so that if we fix the randomness to $r_0$, then Alice and Bob on input $x$ and respectively $y$, compute $f(x,y)$ for $(1-\epsilon)$ fraction of $(x,y)$'s.  The protocol with randomness fixed to $r_0$ is deterministic, and thus has thickness equal to one. Denoting this protocol by $\Pi$ and the transcript of this protocol on input $(x,y)$ by $t(x,y)$, and using Theorem~\ref{t:bound-2} (2), we have
\[
C(t(x,y) \mid \Pi) \geqp C(f(x,y) \mid x, \Pi) + C(f(x,y) \mid y, \Pi),
\]
for $1 - O(\epsilon)$ fraction of $(x,y)$'s. As argued in Example~\ref{e:ex1}, for a fraction of $1- O(\epsilon)$ of $(x,y)$' s
\[
C(f(x,y) \mid x, \Pi) \geq n - (1/2) \cdot \log(1/\epsilon) \mbox{ and }  C(f(x,y) \mid y, \Pi) \geq n - (1/2)\cdot \log(1/\epsilon).
\]
It follows that \[
C(t(x,y) \mid \Pi) \geqp 2n - \log(1/\epsilon),
\]
 for $1 - O(\epsilon)$ fraction of $(x,y)$'s. Since  for every $(x,y)$,
\[
C(t(x,y) \mid \Pi) \leq |t(x,y)| + O(1) \leq k + O(1),
\]
we get that $k \geqp  2n - \log(1/\epsilon) - O(1)$.
\end{example} 
\medskip

\begin{example}
We show that very simple arguments based on inequality~\eqref{e:main} can be used to derive lower bounds for the \emph{approximation} of functions.

The argument below is valid for any group, but for concreteness, let us consider the group $G = (\mathbb{Z}_2^n, +)$. Alice and Bob want to compute $f(x,y) = x+y$, where 
Alice has $x \in \mathbb{Z}_2^n$ and Bob has  $y \in \mathbb{Z}_2^n$. It is not hard to see that the communication complexity (even for nondeterministic communication protocols) is $2n$ because the only monochromatic rectangles are individual cells.

The question is whether Alice and Bob can do better if they are content with an approximation of $x+y$. Formally, let $d$ be a metric with the property that for every 
$u \in \mathbb{Z}_2^n$ and every positive integer $\delta$, the ball $B_\delta (u) = \{v \in \mathbb{Z}_2^n \mid d(u,v) \leq \delta \}$ has cardinality at most $2^{O(\delta)}$ and is effectively enumerable given $u$ and $\delta$.  For example, $d$ can be the Hamming distance. Alice and Bob need a protocol that computes the relation $\Gamma_\delta = \{(x,y,z) \in (\mathbb{Z}_2^n)^3 \mid d(x+y, z) \leq \delta\}$.  Let $(\Pi, t)$ be such a protocol with thickness $\rho(\Pi) = \poly(n)$. By Theorem~\ref{t:bound-2}(2), for every $(x,y)$,
\[
C(t(x,y) \mid \Pi) \geqp C(f(x,y) \mid x, \Pi) + C(f(x,y) \mid y, \Pi). 
\]
Since $x+y \in B_{\delta} (f(x,y))$, it follows that
\[
C(x+y \mid \Pi, x) \leq C(f(x,y) \mid x, \Pi) + O(\delta)  \mbox{  and }  C(x+y \mid y, \Pi) \leq C(f(x,y) \mid y, \Pi) + O(\delta), 
\]
because $x+y$ can be described by its ordinal in $B_{\delta} (f(x,y))$.

By a standard counting argument, for every $c$, for every $x$, and for $(1-2^{-c})$ fraction of $y$'s, $C(x+y \mid x, \Pi) \geq n- c$, and the similar relation holds for 
$C(x+y \mid y, \Pi)$.  It follows that for every $c$, and for $(1 - 2 \cdot 2^{-c})$ fraction of $(x,y)$,
\[
C(t(x,y) \mid \Pi) \geqp 2n - \delta - c.
\]
In words, the communication complexity for computing a $\delta$ approximation of $x+y$ does not drop essentially by more than $\delta$.
\end{example}
\fi

\subsection{Comparison with conventional techniques}

It is instructive to compare the technique discussed above,  based on the inequality~\eqref{e:eqmain} and its variants~\eqref{e:main} and~\eqref{e:main1}, with more standard methods.  Let us consider \emph{nondeterministic} communication protocols
(with no randomness). In this model it is easy to estimate the communication complexity of the function $f(x,y)=x+y$ (for $x,y\in \mathbb{Z}_2^n$).
By definition, a nondeterministic communication protocol can be represented as a collection of ``monochromatic'' combinatorial rectangles that cover the set of all paris of inputs (i.e., the set $ \mathbb{Z}_2^n \times  \mathbb{Z}_2^n$). The property of monochromaticity means that each rectangle should consist of pairs $(x,y)$ with one and the same value of $f(x,y)$. The communication complexity of the protocol is the logarithm of the number of rectangles in the cover.

No two pairs $(x_1, y_1)$ and $(x_2, y_2)$ can be in the same monochromatic rectangle. Indeed, either the pairs do not have the same sum, \ie,  they have different colors; or they have the same sum, but in this case the ``crossed" pairs $(x_1, y_2)$ and $(x_2, y_1)$ have a different sum. Therefore, the number of monochromatic rectangles has to be $2^n \times 2^n$, and thus the communication complexity of the protocol is not less than $2n$. This proof is a version of the fooling set argument. The same bound can be obtained with a more explicit usage  of the standard techniques of fooling sets or linear rank, which have been used to establish lower bounds for many communication problems.
There are $2^n$ possible values of $x+y$, so we must have rectangles colored in each of the $2^n$ possible values. Next, for a fixed value $z$, the relation that consists of all pairs $(x,y)$ such that $x+y=z$ is isomorphic to the relation of identity ($x$ and $z-y$ must be equal to each other). And for this predicate the fooling set or linear rank methods imply that the cover contains at least  $2^n$ rectangles. Summing up the number of rectangles for all $z$ we conclude that the cover consists of at least $2^n \times 2^n$ monochromatic rectangles,

 This argument works for protocols with any thickness. So with a very simple argument we have obtained a statement which is even stronger then the bound $2n - \log(\rho(\Pi)) - O(1)$ that follows from inequality~\eqref{e:eqmain}. However, the conventional techniques are not stable with respect to a random perturbation. When we change the value of $f(x,y)$ for a fraction $\epsilon$ of all pairs $(x,y)$, we can corrupt all large enough fooling sets or dramatically reduce the linear rank. 
On the other hand, the bound based on the information inequality~\eqref{e:eqmain} remains valid (though we need to assume that the thickness of
the protocol is bounded). Thus, the new technique has an advantage if we deal with ``randomly perturbed'' versions of well studied functions. Roughly speaking, with the new argument we gain the factor of $2$ compared to the simpler standard bounds: for protocol with low thickness we obtain a lower bound $\approx 2n$, while the trivial lower bound (the logarithm of the number of colors) is $n$. We used this property of ``robustness'' to prove a lower bound for AM communication protocols in Section~\ref{s:am}.

As an additional example, consider the communication complexity of computing an approximation of $x+y$ (for $x,y\in \mathbb{Z}_2^n$).  Thus, Alice and Bob want to compute  $\fff$ as a $\delta n$-approximation of $x+y$, meaning that the Hamming distance between $\fff$ and $x+y$ is bounded by $\delta n$, for $\delta < 1/2$.   By Theorem~\ref{t:bound-2}, part $2$, for any protocol $\Pi$ computing such an approximation $\fff$ and for all pairs $(x,y)$, we have  $C(\ttt \mid \Pi) \geqp C(\fff \mid x, \Pi) + C(\fff \mid y, \Pi) - \log \rho$. Since $x+y$ can be computed from $\fff$ and $h(\delta)n$ bits, where $h(\delta) =  \delta \log (1/\delta) + (1-\delta) \log (1/(1-\delta))$, the right hand side of the inequality is $\geqp C(x+y \mid x, \Pi) - h(\delta)n + C(x+y \mid y, \Pi) - h(\delta)n - \log \rho$. For most pairs $(x,y)$, $C(x+y \mid x, \Pi) \geqp n$  and $C(x+y \mid y, \Pi) \geqp n$. We conclude  that for most pairs $(x,y)$, $C(\ttt \mid \Pi) \geqp 2(1-h(\delta))n - \log \rho$.  

\begin{remark}
There is another nice property of the technique based on the inequality~\eqref{e:eqmain}: the bound holds true for protocols where the value $f(x,y)$ is not 
embedded explicitly in the transcript $t(x,y)$ of the protocol. As stated in Lemma~\ref{l:bound}, the bound applies to the protocols where Alice and Bob can compute $f(x,y)$ given
the transcript together with their inputs, $x$ or $y$ respectively (so Alice and Bob can find $f(x,y)$, while the external observer who accesses only 
the transcript possibly cannot reconstruct the value of $f$).
\end{remark}

\section{The information complexity of communication protocols}
\label{s:infcomplexity}

In the standard setting of information theory and Shannon entropy, there are two types of information complexity, internal and external.  We focus on the first one, which has more applications, and define the analog concept in the framework of Kolmogorov complexity.
\begin{definition}
The internal information cost of a protocol $\Pi = (\rt, t)$ for input pair $(x,y)$ is
\[
\mathrm{IC}_\Pi (x, y) = I(x:\ttt \mid y, \Pi) + I(y:\ttt \mid x,  \Pi).
\]
\end{definition}
The internal information cost is the amount of  information each party learns about the other party's input from the transcript of the protocol. Intuitively, it should not be more than the complexity of the transcript.  

The next theorem concerns the case of nondeterministic protocol and shows that the internal information complexity is bounded by the sum between the complexity of the transcript and the logarithm of the thickness.  This validates  the intuition mentioned in the above paragraph,  up to  a $O(\log n)$ loss of a precision, for the class of protocols with polynomial thickness, which includes the class of deterministic protocols (because such a protocol has thickness equal to $1$).  The theorem is the Kolmogorov complexity analog of a result of Braverman and Rao~\cite{bra-rao:c:infcomplexity}. We note that the proof of Braverman and Rao cannot be adapted, because they consider only randomized protocols (so, without Merlin) and their proof works inductively on the number of rounds of the Alice/Bob interaction.  In our setting, Merlin also contributes to the communication complexity and this component  is not handled by the technique in~~\cite{bra-rao:c:infcomplexity}, as we have explained in Remark~\eqref{r:remproof}.

\begin{theorem}
For every protocol $\Pi = (\rt,t)$ and every input pair $(x,y)$,
\begin{itemize}
\item[(1)] $\mathrm{IC}_\Pi(x,y) \eqp C(\ttt \mid \Pi) - I(x:y:\ttt \mid \Pi)$.
\item [(2)] $\mathrm{IC}_\Pi(x,y) \leqp C(\ttt \mid \Pi) + \log \rho(\ttt)$.
\end{itemize}
\end{theorem}
\begin{proof}
To keep the notation simple, in all the $C(\ldots)$ and $I(\ldots)$ below we omit the conditioning on $\Pi$. Note that
\begin{equation}
\label{e5}
\begin{array}{ll}
I(x: \ttt \mid y)  & \eqp   C(x \mid y )  - C(x \mid \ttt, y)  \\
& \eqp C(x,y) - C(y) - C(x,\ttt , y) + C(\ttt, y) \\
& \eqp C(\ttt , y) - C(y)  \\
& \eqp  C(\ttt \mid y).
\end{array}
\end{equation}
In the third line we have used the fact that, given $\Pi$, $\ttt$ can be computed from $x$ and $y$.  Similarly,   $I(y:\ttt \mid x) \eqp C(\ttt \mid x)$.

Consequently, the internal information cost $\mathrm{IC}_\Pi(x,y)$ is equal  to $C(\ttt \mid x) + C(\ttt \mid y)$ (up to $O(\log n)$ precision).

Now item (1)  follows from Lemma~\ref{l:mutualf} (2).

Item  (2)  holds because $I(x:y:\ttt)  \eqp I(x:y) - I (x:y \mid \ttt)$ (Lemma~\ref{l:formulas} (1)) and, by Theorem~\ref{t:extension},  the latter expression is at least $- \log \rho(\ttt)$.  

\if01
\begin{equation}
\label{e6}
\begin{array}{ll}
C(t) & \geqp   I (a : b \mid f(a,b))  - I(a:b)  +  C(t \mid  a,f(a,b))  + C(f(a,b)) -  C(f(a,b) \mid b) +  C(t \mid b) \\
 & \hspace{5cm} \mbox{(Eq.~\eqref{e4})} \\
 & \eqp C(f(a,b) \mid a) + C(f(a,b) \mid b) +   C(t \mid  a,f(a,b))  -  C(f(a,b) \mid b) +  C(t \mid b) \\
  & \hspace{5cm} \mbox{(by Lemma~\ref{l:mutualf})} \\
  & \eqp C(f(a,b) \mid a) + C(t \mid a, f(a,b)) + C(t \mid b) \\
  &\eqp C(a, f(a.b)) - C(a) + C(t, a, f(a,b)) - C(a, f(a,b)) + C(t \mid b) \\
  & \eqp - C(a)  +  C(t, a) + C(t \mid b) \\
  &\eqp C(t \mid a) + C(t \mid b).
\end{array}
\end{equation}
\fi

\end{proof}

\section{Lower bounds for multi-party protocols}
\label{s:multiparty}
The results in the  Section~\ref{s:twoparty} can be extended to protocols involving more than two parties, in the so called \emph{number-in-hand} model. In this model, there are $\ell  \geq 2$ parties, $P_1, \ldots, P_\ell$ and a function $f$ of $\ell$ arguments. For every $i \in [\ell]$, party $P_i$ has input $x_i$, and the goal is for every $P_i$ to compute the value $f(x_1, \ldots, x_\ell)$ via some communication protocol.  The basic considerations in Section~\ref{s:twoparty} remain valid, modulo straightforward adjustments. In particular, the set of inputs on which the protocol has a given transcript $t$ is again a rectangle, where this time a rectangle is a set $R$ of the form $R= S_1 \times \ldots \times S_\ell$.  Therefore, deterministic protocols induce a partition of the domain of $f$ into monochromatic rectangles, and nondeterministic protocols induce a covering of the domain of $f$ with monochromatic rectangles. Definitions~\ref{d:protocol},~\ref{d:thick},  and~\ref{d:protocolcompute} have straightforward analogs for the case of $\ell$ parties.

The following is a generalization of Theorem~\ref{t:extension} and the proof is similar.
\begin{theorem}
\label{t:extension-multi}
For every protocol $\Pi = (\rt, t)$, and for every $(x_1, \ldots, x_\ell) \in D$, 
\[
C(\ttx \mid \Pi) \geqp \frac{1}{\ell - 1} \sum_{i=1}^\ell C(\ttx \mid x_i, \Pi)  - \frac{1}{\ell - 1}  \log \rho(\ttx),
\]
where $\ttx$ is an abbreviation for $t(x_1, \ldots, x_\ell)$.
\end{theorem}

The following is a generalization of Theorem~\ref{t:bound-1} and the proof is similar.
 \begin{theorem} 
\label{t:bound-multi}
For every computable function $f$ of $\ell$ arguments, for every protocol $(\Pi, t)$ that computes $f$ over domain $D$, and for every $(x_1, \ldots, x_\ell) \in D$,
\begin{equation}
\label{e:main2}
C(\ttx \mid \Pi) \geqp \frac{1}{\ell-1}\bigg( \sum_{i=1}^{\ell} C(\ffx \mid x_i, \Pi) +   \sum_{i=1}^{\ell} C(\ttx \mid x_i, \ffx, \Pi)-  \log \rho(\ttx)
  \bigg),
\end{equation}
where $\ttx$ is an abbreviation for $t(x_1, \ldots, x_\ell)$ and $\ffx$ is an abbreviation for $f(x_1, \ldots, x_\ell)$.
\end{theorem}
\begin{example}
We present a lower bound for $\ell$-party protocols $\Pi$ which compute the following function $f$.  Here $\ell$ is a constant and the inputs $x_1, \ldots, x_{\ell-1}$ are $n$ dimensional column vectors and $x_\ell$ is an $\ell-1$ dimensional column vector  over $\mathbb{Z}_2$. The function $f(x_1, \ldots, x_{\ell-1}, x_\ell)$ is defined as $A \cdot x_\ell$ where $A$ is the $n$-by-$(\ell-1)$ matrix having the columns $x_1, \ldots, x_{\ell-1}$. 
Since $A \cdot x_\ell$ is $n$ bits long, there are  $2^n$ ``colors" and therefore the number of $f$-monochromatic rectangles is at least $2^n$.  Using Theorem~\ref{t:bound-multi} we obtain a better lower bound.  If $n$ is sufficiently large, $A$ has rank $\ell-1$ with high probability when its columns are chosen uniformly at random. It can be shown that with probability $1-O(\epsilon)$, $C(f(x_1, \ldots, x_\ell) \mid x_i, \Pi) \geq n - \log(1/\epsilon)$ for every $i \in [\ell]$.
Now it follows from inequality~\eqref{e:main2} that with probability $1-O(\epsilon)$
\begin{equation*}
\begin{array}{ll}
C(t(x_1, \ldots, x_\ell) \mid \Pi) & \geqp \frac{1}{\ell-1}\bigg( \sum_{i=1}^{\ell} C(f(x_1, \ldots, x_\ell) \mid x_i, \Pi)  - \log(\rho_{x_1, \ldots, x_\ell}) 
 \bigg)\\
& \geqp \frac{\ell}{\ell-1} (n-\log(1/\epsilon)) - \frac{1}{\ell-1}  \log \rho(\ttx).
\end{array}
\end{equation*}
For example, if $\ell = 3$, the above argument shows that in any protocol with $2^{o(n)}$- thickness Merlin needs to send at least approximately $(3/2) n - o(n) - O(1/\epsilon)$ bits for a fraction of $(1-\epsilon)$ of input tuples, whereas the simple argument based on the number of ``colors" only gives a lower bound of $n$ and furthermore this lower bound is in the worst-case sense.
\end{example}


\bibliography{theory-3}

\bibliographystyle{alpha}

\end{document}